\documentclass{article}
\usepackage{enumitem}
\usepackage{tikz}
\usetikzlibrary{matrix,arrows,calc,automata}
\usepackage{pgfplots}
\pgfplotsset{compat=newest}
\usepackage{booktabs}
\usepackage{color}
\usepackage[cmex10]{amsmath}
\usepackage{cleveref}
\usepackage{amssymb,euscript,psfrag,latexsym,graphicx}
\usepackage{bbm,color,amstext,wasysym,cuted,mathtools, cite}
\usepackage{algpseudocode}

\newtheorem{theorem}{Theorem}
\newtheorem{definition}{Definition}
\newtheorem{lemma}{Lemma}
\newtheorem{proof}{Proof}

\newtheorem{remark}[theorem]{Remark}
\newtheorem{property}{Property}

\begin{document}

\title{Evaluation of the Region of Attractions of Higher Dimensional Hyperbolic Systems using the Extended Dynamic Mode Decomposition}

\date{}
\author{{\small C. Garcia-Tenorio$^{*,1,2}$, D. Tellez-Castro$^{1}$, E. Mojica-Nava$^{1}$, A. Vande Wouwer$^{2}$}\thanks{Camilo.garciatenorio@umons.ac.be}\\
{\footnotesize $^1$, Universidad Nacional de Colombia, Carrera 30 No. 45--03, Bogotá, Colombia}\\
{\footnotesize $^2$ Systems, Estimation, Control, and Optimization (SECO), Universit\'e de Mons,7000 Mons, Belgium}}

\maketitle
\abstract{%
   This paper proposes an original methodology to compute the regions of attraction in hyperbolic and polynomial nonlinear dynamical systems using the eigenfunctions of the discrete-time approximation of the Koopman operator given by the extended dynamic mode decomposition algorithm. The proposed method relies on the spectral decomposition of the Koopman operator to build eigenfunctions that capture the boundary of the region of attraction. The algorithm relies solely on data that can be collected in experimental studies and does not require a mathematical model of the system. Two examples of dynamical systems, a population model and a higher dimensional chemical reaction system, allows demonstrating the reliability of the results.
}
\footnotetext{\textbf{Abbreviations:} EDMD, extended dynamic mode decomposition; ROA, region of attraction}

\section{Introduction}%
\label{sec:introduction}
The analysis of nonlinear systems often focuses on the stability of the fixed point or equilibrium point of the system, especially the global stability of this unique point. When the fixed point is not unique, the concept of the region of attraction (ROA) is as important as the stability of the point. In general terms, the ROA represents the extent to which a disturbance can drive the system away from a stable equilibrium point so that it can still return to it. In other words, the ROA indicates from which initial conditions the system converges to a stable point. This analysis is particularly useful in biological or ecological systems, where there are several equilibrium points in which one or all the species vanish, and some points where the species coexist~\cite{Chellaboina2009a,Bomze1983}. 

There are some traditional techniques to approximate the ROA of asymptotically stable equilibrium points, such as the computation of the level sets of a local Lyapunov function~\cite{Khalil2015}, the backward integration of systems from saddle equilibrium points to approximate the boundary of the ROA~\cite{chiang2015stability}, or the computation of the level sets of a local energy function, similar to the Lyapunov approach. As local energy functions are constant along the trajectories of the system, they can provide the stable manifold of saddle points in the boundary. Consequently, these stable manifolds form the whole boundary of the ROA~\cite{chiang2015stability}. Among these techniques, integrating the system back from the saddle points in the boundary does not require the knowledge of the local Lyapunov (or energy) function of the system, and gives a less conservative approximation.

All the above methods rely upon the explicit knowledge of a mathematical model, i.e., a set of differential equations that result from mathematical modeling and parameter identification methods~\cite{Garnier2008,Augusiak2014,ElKalaawy2015}, and the backward and forward integration require that the system is backward integrable. 

In this paper, we propose a data-driven approach, solely based on information gathered either from the collection of experimental data from a physical system or the numerical integration of an existing numerical simulator of arbitrary form and complexity. If experimental data is available, there is no need to derive a mathematical model and to identify the model parameters. This data, along with numerical methods, allows the approximation of the Koopman operator, which has a set of eigenfunctions that can be used for the approximation of the ROA~\cite{Koopman1931,Williams2015}.  

Rather than describing the time evolution of the system states, the Koopman operator describes the evolution of eigenfunctions~\cite{Mezi2017}. Analyzing these eigenfunctions allows the identification of invariant subspaces that acquire specific characteristics of the dynamical system. For example, if an eigenfunction has an associated eigenvalue equal to one, the value of this function is invariant along the trajectories of the dynamical system. As a consequence, an eigenfunction with an associated eigenvalue equal to one defines an invariant subspace useful to capture specific characteristics of the dynamical system.

Some recent studies~\cite{Mezic2005a,Lan2013,Mauroy2013,Mauroy2016} point out that eigenfunctions that have an associated eigenvalue equal to one are useful for nonlinear system analysis. However, they do not formalize theoretically the characteristics of these functions to perform the analysis. In \cite{Mezic2005a}, the authors use time-averages of observables, i.e., the average of arbitrary functions of a trajectory of the state from a particular initial condition. These functions converge to the unitary (or near unitary) eigenfunction, and therefore their level sets provide visual information on the state space partition. This procedure provides some insight into the partition of the state space in two-dimensional systems or slices of three-dimensional ones. However, this analysis does not provide a criterion to get the partition of the state space, as it only gives visual information, and this is the reason it is limited to low dimensional systems. In addition, for the calculation of the time averages, it is necessary to have the solution (or numerical integration) of the difference/differential equations from every point of the state space under consideration to get the graph depicting the level sets. Another notable approach comes from \cite{Mauroy2013}, that uses the isostables of a system. An isostable of a stable equilibrium point is a set of initial conditions whose trajectories converge  synchronously to the attractor, that is, their trajectories simultaneously intersect subsequent isostables along a trajectory that converges to a stable point. The definition of an isostable comes from the magnitude of the slowest Koopman operator eigenfunction, whose level sets give the isostables. Another important contribution in~\cite{Mauroy2016} is a global stability analysis and the approximation of the region of attraction based on the spectral analysis of the Koopman eigenfunctions. The global stability analysis comes from the zero level sets of the Koopman eigenfunctions related to Koopman eigenvalues whose real part is less than zero. The approximation of the region of attraction comes from the traditional analysis of Lyapunov functions, which are built upon the same Koopman eigenfunctions and eigenvalues from the global stability analysis. To get these eigenfunctions and eigenvalues, the authors use a numerical method based on Taylor and Bernstein polynomials, where it is necessary to know the analytic vector field of the system. Similar to the Lyapunov and energy function-based methods, these methods require the calculation of the level sets of a particular function, which yields the same problems regarding the classification of an arbitrary initial condition in higher dimensional spaces. 

Probably the most notable contribution that comes closer to developing a purely data-driven technique comes from Williams et al., with the development of the extended dynamic mode decomposition algorithm~\cite{Williams2015}. The authors provide more insight into the determination of the ROA for the particular case of a Duffing oscillator with two basins of attraction, and they analyze the leading (unitary) eigenfunction to determine which basin of attraction a point belongs to. The criterion for this classification is the use of the mean value of the unitary eigenfunction as a threshold. Subsequently, the authors parametrize the ROA by recalculating an approximation of the Koopman operator on each basin. Even though the development in~\cite{Williams2015} shows an accurate procedure to get the ROA in the particular case of the Duffing oscillator, there are some issues not covered in the development. For instance, there is no guarantee of the existence of unitary eigenfunctions coming from the approximation methods of the Koopman operator. When a unitary eigenfunction is present, it is often trivial, i.e., the function value is constant in all the state space, and as a consequence, it does not provide information about the system. Furthermore, the ROA analysis methods based on the fixed points of the system rely on the model equations and their linearization to determine the fixed point location and stability. Moreover, the mean value of the unitary eigenfunction may work on the particular case of the Duffing oscillator, but it is not a generalized criterion. Hence, it may not work for other types of systems and there is a lack of guidelines on how to generalize to higher dimensional systems.  

The main contributions of this paper are the solutions to many of the aforementioned problems. The data-driven restriction is held throughout the development, meaning that all the necessary information comes from the approximation of the Koopman operator, including the location and local stability of the fixed points. In addition, our approach presents a suitable algebraic condition to determine the ROA given a set of unitary eigenfunctions, devoid of level set calculation or complex geometric analysis to get the classification of a particular initial condition. As a consequence of having an algebraic condition, our approach is suitable for analyzing higher-dimensional dynamical systems (with dimensionality higher than three). 

In contrast, some shortcomings of our approach are the fact that we are working with simulated data exempt of noise and the assumption of full observability of the system, i.e., complete knowledge of the state.

This paper is organized as follows. The next section introduces the main concepts of region of attraction, Koopman operator theory and the extended dynamic mode decomposition algorithm. Section~\ref{sec:Results} presents and discusses the main contribution of this study, i.e., a numerical procedure to evaluation the regions of attractions based on the EDMD and the calculation of eigenfunctions associated to unitary eigenvalues. Next, The methodology is applied to two examples in Section~\ref{sec:simulation_results}, e.g., a population model, and a higher-dimensional chemical reaction system. In both cases, the data is provided by the numerical simulation of a model, but we stress that the procedure is applicable to experimental data as well, provided that it is available in sufficient numbers to secure a sufficiently accurate approximation. The last section is devoted to concluding remarks and future prospects.

{\textbf Notation} $A^{\top}$ and $A^{+}$ are the transpose and pseudoinverse of a matrix, $A\in\mathbb{R}^{n\times n}$ respectively. 

\section{Basic Concepts and Methods}
\label{sec:methods}
\subsection{Regions of Attraction}
\label{sub:ROA}
Consider a nonlinear dynamical system $(\mathcal{M};T(x);k)$ in discrete-time, with state variables $x\in\mathcal{M}$ where $\mathcal{M}\subseteq\mathbb{R}^{n}$ is the nonempty compact state space, $k\in\mathbb{Z}_{0}^{+}$ is the discrete time, and $T\colon{}\mathcal{M}\rightarrow{}\mathcal{M}$ is the differentiable vector-valued evolution map, i.e.,
\begin{equation}
     x(k+1)=T(x(k)),\quad{}x_0=x(0).
     \label{eq:DiscreteDS}
 \end{equation} 

The solution to~\eqref{eq:DiscreteDS} is the successive application of $T$ from an initial condition $x_0\in{}\mathcal{M}$ at $k=0$, i.e., $x_k=T^k(x_0)\in\mathcal{M}$, which is an infinite sequence  called a trajectory of the system. Suppose $x^{*}\in{}\mathcal{M}$ is a fixed point of~\eqref{eq:DiscreteDS}; i.e., 
\begin{equation}
    T^{k}(x^{*})=x^{*}.
    \label{def:FixedPoint}
\end{equation}

The linearization principle defines the local stability of hyperbolic fixed points, i.e., points that satisfy the Hartman-Grobman theorem~\cite{Khalil2015,Coayla-Teran2007}. This principle states that a fixed point $x_{s}^{*}$ is asymptotically stable if the modulus of all the eigenvalues of the Jacobian matrix evaluated at the fixed point are less than one, and unstable otherwise. Additionally, the {\it type-k\/} of a hyperbolic fixed point is defined as the number $k$ of eigenvalues with modulus greater than one. If only one eigenvalue has modulus greater than one, the fixed point $x^*$ is a {\it type-one\/} point. When the index $k$ of unstable fixed points is equal or greater than one, and less than $n$, the fixed point is called a saddle, denoted by $\hat{x}^{*}$. The {\it type-one\/} saddle points play an important role in the approximation of the ROA.

Eigenvalues $\lambda$ have a corresponding eigenvector $E_{\lambda}$ that spans the eigenspace associated with the eigenvalue. For the eigenvalues with modulus smaller than one, the direct sum of their eigenspaces is the generalized stable eigenspace of a fixed point, i.e., $E^s=\oplus{}E_{|\lambda|<1}$. Conversely, for eigenvalues with modulus greater than one, the direct sum of their eigenspaces is the generalized unstable eigenspace of a fixed point, i.e., $E^u=\oplus{}E_{|\lambda|>1}$. The type of the hyperbolic fixed point defines the dimension of the corresponding eigenspaces, $E^u\in\mathbb{R}^k$ and $E^s\in\mathbb{R}^{n-k}$. The state space $\mathbb{R}^n$ is the direct sum of the two invariant stable and unstable eigenspaces $\mathbb{R}^n=E^s\oplus{}E^u$.

As the Hartman-Grobman theorem establishes a one-to-one correspondence between the nonlinear system and its linearization, locally, the stable and unstable eigen\-spaces are tangent to the unstable and stable manifolds of the hyperbolic fixed point. The definitions of these manifolds are
\begin{align}
    W^{u}(x^{*})&=\{x\in\mathbb{R}^{n}\colon{}\lim_{k\rightarrow-\infty}T^{k}(x)=x^{*}\},\label{eq:unstableM}\\
    W^{s}(x^{*})&=\{x\in\mathbb{R}^{n}\colon{}\lim_{k\rightarrow\infty}T^{k}(x)=x^{*}\},\label{eq:stableM}
\end{align}
for the unstable and stable manifold of a fixed point respectively, assuming that an inverse for the backward flow exists for $T^{k}$.  

Before the statement of the theorem characterizing the ROA of an asymptotically stable point, denote the ROA of the fixed point as $R_{A}(x^{*}_{s})$ and the stability boundary as $\partial{}R_{A}(x^{*}_{s})$. From the definitions of the unstable and stable manifolds~\eqref{eq:unstableM} and~\eqref{eq:stableM}, any system under analysis must satisfy the following three assumptions: 
\begin{enumerate}[label=A\arabic*:]%
  \item All the fixed points on $\partial{}R_{A}(x^{*}_{s})$ are {\it type-one}.
  \item {\sloppy The $W^{u}(x^{*})$ and $W^{s}(x^{*})$ of the {\it type-one\/} points on $\partial{}R_{A}(x^{*}_{s})$ satisfy the transversality condition.}
  \item Every trajectory that starts on $\partial{}R_{A}(x^{*}_{s})$ converges to one of the {\it type-one\/} points as $k\rightarrow\infty$.
\end{enumerate}
\begin{remark}%
Manifolds $A$ and $B$ in $\mathcal{M}$ satisfy the transversality condition if the intersection of the tangent spaces of A and B span the tangent space of $\mathcal{M}$.
\end{remark}

For any hyperbolic dynamical system that satisfies assumptions (A1-A3), the region of attraction of an asymptotically stable point is,
\begin{theorem}%
\label{th:AttractionRegions}{~\cite[Th.~9-(10,11)]{chiang2015stability}}%
    Consider the dynamical system~\eqref{eq:DiscreteDS} and assume it satisfies assumptions (A1-A3).
    Define ${\{\hat{x}^{*}_i\}}_{i=1}^{P}$ as the $P$ {\it type-one\/} hyperbolic fixed points on the boundary, the stability region of an asymptotically stable fixed point. Then,
    \begin{enumerate}
        \item $\hat{x}^{*}_{i}\in\partial{}R_A(x_{s})$ $\iff$ $W^{u}(\hat{x}^{*}_i)\cap{}R_A(x^{*}_{s})\neq\varnothing$
        \item $\partial{}R_A(x^{*}_{s})=\cup{}W^{s}(\hat{x}^{*}_i)$.
    \end{enumerate}
\end{theorem}
\noindent Hence, the stability boundary is the union of the {\it type-one\/} hyperbolic fixed points stable manifolds on the stability boundary.
\begin{remark}\label{rem:assumptions}
Assumption (A1) is a generic property of differentiable dynamical systems, while assumptions (A2) and (A3) must be verified. 
\end{remark}
\subsection{Basics of Koopman Operator Theory}%
\label{sub:KO}
Consider a set of arbitrary functions of the state, the so-called {\it observables\/} $f(x)\colon\mathcal{M}\rightarrow{}\mathbb{C}$ for system~\eqref{eq:DiscreteDS} that belong to some function space, i.e., $f(x)\in\mathcal{F}$. There is a discrete-time linear operator $U^{k}$, the Koopman operator, which defines the time evolution of these observables, i.e., 
\begin{equation}
    \left[U^{k}f\right](x)=f\left(T^{k}(x)\right).
    \label{eq:Koopman}    
\end{equation}
The left-hand side of~\eqref{eq:Koopman} is the time evolution of the observables, while the right-hand side is the time evolution of the state subsequently evaluated by the observables. The Koopman operator is linear but infinite dimensional, and some form of truncation to a finite dimensional approximation will be required in practice, introducing a trade-off between accuracy and dimensionality.

The linear operator has a spectral decomposition of tuples ${\{(\mu_i,\phi_i(x),v_i)\}}_{i=1}^{\infty}$ of eigenvalues, eigenfunctions, and modes. The eigenvalues and eigenfunctions satisfy the condition that the corresponding eigenvalue determines the dynamics or time evolution of a specific eigenfunction
\begin{equation}
    [U^{k}\phi_{i}](x)=\mu_{i}^k\phi_{i}(x),
    \label{eq:PhiEvol}
\end{equation}
and the modes map the linear evolution of eigenfunctions~\eqref{eq:PhiEvol} into the original observables by weighting the eigenfunctions  
\begin{equation}
    f(x)=\sum_{i=1}^{\infty}v_{i}\phi_{i}(x).
    \label{eq:FullState}
\end{equation}
The importance and advantages of the Koopman operator and the diagonalization provided by the spectral decomposition are highlighted by equations ~\eqref{eq:PhiEvol} and~\eqref{eq:FullState}. Hence, the evolution of observables with respect to the spectral decomposition of the Koopman operator is
\begin{equation}
    f\left(T^k(x)\right)=\left[U^{k}f\right](x)=\sum_{i=1}^{\infty}v_i\mu_i^k\phi_i(x).
    \label{eq:ObsEvol}
\end{equation} 

\subsection{Extended Dynamic Mode Decomposition Algorithm}
\label{sub:EDMD}
The objective of the EDMD algorithm~\cite{Williams2015} is to get a finite and discrete-time approximation of the Koopman operator based on sampled data of the underlying dynamical system~\cite{Klus2016,Korda2017a}.  

The EDMD algorithm that approximates the discrete-time Koopman operator of~\eqref{eq:DiscreteDS} requires $N$ pairs of snapshot data, either from the numerical solution of an existing mathematical model or from experimental measurements at a specific sampling time $\Delta{}t$. The sets of snapshot pairs ${\{(x_i,y_i)\}}_{i=1}^{N}$ satisfy the relationship $y_i=x_{i+1}=T(x_i)$ and their definition in matrix form is 
\begin{align}
    X&=\begin{bmatrix}x_1&x_2&\cdots{}&x_N\end{bmatrix},&Y&=\begin{bmatrix}y_1&y_2&\cdots{}&y_N\end{bmatrix}.
\end{align}
The {\it extended\/} part of the EDMD algorithm consists in the approximation of the Koopman operator on a {\it lifted\/} space of the state variables, rather than approximating the state space as in the dynamic mode decomposition algorithm~\cite{SCHMID2010}. The {\it lifting\/} procedure consists in evaluating the state of the system with the set of observables $\Psi{}\colon{}\mathcal{M}\rightarrow\mathbb{C}^{d\times{}1}$; $\Psi(x)={[\psi_1(x),\ldots,\,\psi_d(x)]}^\top$. However, the choice of observables for a particular system is still an open question, and the common choices are orthogonal polynomials~\cite{Koekoek2010}, radial basis functions, or an arbitrarily constructed set with polynomial elements, trigonometric functions, logarithmic functions or any combination of them~\cite{Brunton2015}. Our choice is a low-rank orthogonal polynomial basis~\cite{Konakli2016a,Konakli2016,Garcia-Tenorio2021a,Garcia-Tenorio2021b}, where every element of $\Psi{}(x)$ is the tensor product of $n$ univariate orthogonal polynomials from a set ${\{\pi_{\alpha_j}(x_{j})\}}_{\alpha_j=0}^{p}$, where $\alpha_j$ is the degree of the polynomial on the $j^{\text{th}}$ component of the $x$ vector, and $p$ is the maximum degree of the polynomial. Every component of $\Psi{}(x)$ is given by
\begin{equation}
    \psi_l(x)=\prod_{j=1}^{n}\pi_{\alpha_j}(x_{j}).
\end{equation}
The low-rank orthogonal polynomial basis comes from the truncation scheme for the non-empty finite set of indices $\alpha$, where the choice of indices is based on $q$-quasi-norms\footnote{The quantity $\Vert\cdot\Vert_{q}$ is not a norm because it does not satisfy the triangle inequality.}; $\alpha=\{\alpha\in\mathbb{N}^{n}\colon{}\Vert\alpha\Vert_{q}\leq{}p\}$ with
\begin{equation}
    \Vert\alpha\Vert_{q}={\left(\sum_{i=1}^{n}\alpha_i^q\right)}^{\frac{1}{q}}.
\end{equation} 

This choice of truncation scheme has two advantages. First, it reduces the number of elements in the set of observables $\Psi(x)$, and therefore, handles the curse of dimensionality problem when the dimension of the state space grows and the available computational resources are limited. The second advantage is empirical and refers to the numerical stability of the least squares solution for the approximation of the Koopman operator. Having a reduced orthogonal basis improves the condition number of the matrices involved in the computation. As a consequence, the algorithm does not rely on the pseudo inverse of a matrix to get the approximation and its accuracy also increases. For a detailed description of the use of p-q-quasi norms for the EDMD algorithm, we refer to previous works by the authors~\cite{Garcia-Tenorio2021a,Garcia-Tenorio2021b}.

Furthermore, the approximation of the discrete-time $d$-dimensional Koopman operator $U_d$, satisfies the following condition~\cite{Williams2015}
\begin{equation}
    \Psi(Y)=U_d\Psi(X)+r(X),
    \label{eq:Kcond}
\end{equation}
where $r(X)\in\mathcal{F}$ is the residual term to minimize in order to find $U_d$. This minimization accepts a closed form solution within the least mean squares problem, where the objective function has the form 
\begin{equation}
    \Vert{}r(X)\Vert^2=\frac{1}{N}\sum_{i=1}^{N}\frac{1}{2}\left\Vert\Psi(y_i)-U_{d}\Psi(x_i)\right\Vert_2^2,
    \label{eq:min}
\end{equation}
and the solution is 
\begin{equation}
    U_d\triangleq{}AG^{+},
\end{equation}
where $G,A\in\mathbb{C}^{d\times{}d}$ are square matrices given by
\begin{align}
    G=&\frac{1}{N}\sum_{i=1}^{N}\Psi(x_i){\Psi(x_i)}^{\top},\label{eq:Gmatrix}\\
    A=&\frac{1}{N}\sum_{i=1}^{N}\Psi(y_i){\Psi(x_i)}^{\top}.\label{eq:Amatrix}
\end{align}

The finite-dimensional and discrete-time approximation of the Koopman operator from the EDMD algorithm has eigenvalues $M=\text{diag}(\mu_1,\,\ldots,\,\mu_{d})$, right eigenvectors $\Xi=[\xi_1,\ldots,\xi_{d}]$, and left eigenvectors, $\Xi^{-1}=W^{\star}$. Furthermore, the approximation of the eigenfunctions $\Phi={[\phi_1,\ldots,\phi_d]}^{\top}$ comes from weighting the set of observables with the matrix of left eigenvectors~\cite{Williams2015,Klus2016}: 
\begin{equation}
    \Phi^\top(x)={\Psi(x)}^{\top}W^{\star}.
\label{eq:PHI}
\end{equation}

Using~\eqref{eq:FullState}, the recovery of the original observables $\Psi(x)$ from the set of eigenfunctions is provided by:
\begin{equation}
     \Psi(x)=\Xi\Phi(x).
     \label{eq:gdPhi}
\end{equation} 

Using~\eqref{eq:ObsEvol}, the time evolution of observables according to the spectrum of the Koopman operator is given by: 
\begin{equation}
      \Psi(T^k(x)) =\Xi{}M^k\Phi(x).
      \label{eq:ObsEvolEDMD}
\end{equation}  
The common practice to recover the state is to include the functions that capture each of the states in the set of observable, i.e., $\psi_i(x)=x_i$. The value of the time-evolution of the states is the value of these particular elements of $\Psi(x)$. The matrix $B\in\mathbb{R}^{d\times{}n}$ recovers these values from~\eqref{eq:ObsEvolEDMD}, where $B$ is a matrix of unit vectors $\hat{e}_i$ indexing the $n$ observables that capture every one of the states, i.e., $f_i(x)=x_i$. As for a polynomial basis, these elements are not always present, and $B$ is a matrix of unit vectors indexing injective observables. To clarify the importance of using the inverse of injective observables against other methods to recover the state, we refer to previous works by the authors~\cite{Garcia-Tenorio2021a,Garcia-Tenorio2021b}.

The evolution of observables~\eqref{eq:ObsEvolEDMD} is used to define the state evolution map according to the approximation of the Koopman decomposition ${\{(\mu_i,\phi_i(x),v_i)\}}_{i=1}^{d}$, i.e., evolving~\eqref{eq:DiscreteDS} $k$ times from an initial condition $x_{0}$
\begin{equation}
    x(k)=T^{k}(x_{0})=\Psi_{B}^{-1}\left(B^\top\Xi{}M^k\Phi(x_0)\right),
    \label{eq:ForwardFlow}
\end{equation}
where $\Psi_B^{-1}(x)$ denotes the inverse of the injective observables selected by $B$.

Note that in addition to the forward evolution of the states, the discrete-time approximation also gives their evolution in reverse time
\begin{equation}
    x(-k)=\Psi_{B}^{-1}\left(B^\top\Xi{}M^{-k}\Phi(x_0)\right).
    \label{eq:BackwardFlow}
\end{equation}

In conclusion, the EDMD allows for a linear discrete-time representation on an extended space of a nonlinear system, with the advantages that the induced spectrum and evolution maps provide. These advantages are the core of the analysis of the ROA. 

For completeness in stating not only the advantages, but also the shortcomings of the approximation, note that the EDMD algorithm assumes the knowledge of the full state of the system and does not take noisy signals into account. Therefore, the application of the algorithm in a real case scenario requires complementing the development with an observer able to handle these issues.

\section{Evaluation of the ROA using EDMD}
\label{sec:Results}
For obtaining the ROA of asymptotically stable fixed points in a data-driven fashion using the EDMD approximation of the Koopman operator, it is necessary to accomplish numerous tasks. First, it is necessary to get the location of the fixed points and determine their stability, especially for the asymptotically stable and the {\it type-one\/} points. Subsequently, the spectrum of the Koopman operator approximation can be analyzed along with the theoretical concepts from \cref{sub:ROA} to approximate the ROA of the asymptotically stable fixed points under the multistability phenomena.

In order to accomplish this objective, several assumptions must hold for the application of the algorithm. First, the EDMD algorithm has knowledge of the full state and enough trajectories from the dynamical system to have an accurate approximation of the operator. This condition can be checked with the empirical error that compares the orbits of the system from~\eqref{eq:DiscreteDS} with the discrete-time Koopman state evolution map~\eqref{eq:ForwardFlow}. Second, assumptions (A1--A3) introduced in subsection \cref{sub:ROA} hold. In general, however, it is necessary to know the differential equation model to check assumptions (A2) and (A3). Due to this difficulty, we limit the analysis to mass action systems for the illustration of the results in \cref{sec:simulation_results}. This type of systems inherently satisfies the assumptions~\cite{Chellaboina2009a}.

\subsection{Fixed Points Approximation}
\label{sub:equilibrium_points}
The first step of the analysis consists in locating the fixed points of the hyperbolic system to assess their stability. A fixed point of the discrete time nonlinear system~\eqref{eq:DiscreteDS} is an invariant subspace that has the property that whenever the state starts in it, it will remain in it for all future time, as in~\eqref{def:FixedPoint}.

To approximate the fixed points, consider the forward and backward state evolution maps of a dynamical system in discrete-time given by the Koopman operator approximation in~\eqref{eq:ForwardFlow} and~\eqref{eq:BackwardFlow}. If the evolution of these mappings is invariant, i.e., maps to themselves, then the system is at a fixed point. To accurately approximate these points, we propose to formulate a minimization problem where the objective function is the euclidean norm of the comparison between the state and its approximate discrete-time evolution~\eqref{eq:ForwardFlow}~\cite{Garcia-Tenorio2019}. 

\begin{lemma}[Fixed points]%
\label{lm:EquilibriumPoints}
    Let $(\mathcal{M};T(x);k)$ be a dynamical system that admits a Koopman operator approximation $(\mathcal{F}_d;U_{d};k)$, and $x^{*}$ is a fixed point of $T(x)$. Then
    \begin{equation}        x^{*}=\min_{x}\left\Vert\Psi_{B}^{-1}\left(B^\top\Xi{}M^k\Phi(x_0)\right)-x\right\Vert^2_2.
        \label{eq:Equilibria}
    \end{equation}
\end{lemma}
\begin{proof}
    Let $T^{k}(x)=\Psi_{B}^{-1}(B^\top\Xi{}M^k\Phi(x_0))$ the flow map of $(\mathcal{M};T(x);k)$, assume $k=K$ finite, and define $b(x)=\Psi_{B}^{-1}(B^\top\Xi{}M\Phi(x))$, using Definition~\ref{def:FixedPoint} and $b(x)$, the least squares problem
    \begin{equation}
        J(x)=\frac{1}{2}\left\Vert{}b(x)-x\right\Vert^2,
    \end{equation}
    with solution
    \begin{equation}
        x^*=\underset{x}{\text{argmin}}\left(J(x)\right),
    \end{equation}
    gives the location of the fixed points.
\end{proof}
\begin{remark}
    Note that this procedure is possible in the portion of the state space from where the data snapshots lie, and corresponds to the fixed points of the nonlinear underlying hyperbolic dynamical system. If the system is not hyperbolic, this condition does not hold. 
\end{remark}
\begin{remark}
    This lemma is only accountable for the location of the fixed points in the state space; it does not give information about their stability.
\end{remark}
\begin{remark}
   Equation~\eqref{eq:ForwardFlow} is a nonlinear, discrete-time evolution mapping, and under the definition of fixed point, i.e., $\bar{T}^k(x^*)=x^*$ it is possible to get the approximation of the hyperbolic fixed points by solving
   \begin{equation}
     \bar{x}^*\subseteq\{x\in\mathbb{R}^n\colon\bar{T}^k(x)-x=0\}.
   \end{equation}
   In practice, this procedure has two issues. First, as the nonlinear mapping~\eqref{eq:ForwardFlow} comes from a set of observables, the dimension and complexity of these functions affect the possibility of getting a solution in polynomial time. As a consequence, in higher order polynomial expansions, even if it is possible to find a solution, the computational time is high.
   Second, when there is a solution, it is not clear which elements of the subset represent an actual fixed point of the system. Solving a nonlinear set of equations gives several solutions in the complex plane where some of which correspond to the real-valued solutions of the system, i.e., not all the points of the subset are fixed points of the system, while the converse is true. Considering the data-driven case, the fixed points of the system are not known a priori and therefore using this definition to approximate the fixed points is not feasible.
 \end{remark}

With the location of the fixed points, it is necessary to establish their stability.

\subsection{Stability of Fixed Points}
\label{sub:stability_of_equilibria}
The traditional way of establishing the local stability of a hyperbolic fixed point is through the analysis of the Jacobian matrix of the system~\eqref{eq:DiscreteDS} evaluated at the fixed point. Our proposed approach is to analyze the state evolution map from the discrete approximation of the Koopman operator in the same way. In other words, the stability comes from the linearization of~\eqref{eq:ForwardFlow} which is a nonlinear mapping of the state.

\begin{definition}[Stability] 
   Let $(\mathcal{M};T(x);k)$ be a dynamical system that admits a Koopman operator approximation $(\mathcal{F}_d;U_{d};k)$. Define $b(x)=\Psi_{B}^{-1}(B^\top\Xi{}M^k\Phi(x))$ from the state evolution map~\eqref{eq:ForwardFlow}, the local linearization of $b(x)$ at a given point $x^{*}$ is
\begin{align}
      x(k+1)=&\left.\left[\frac{\partial{}T(x)}{\partial{}x_1}\cdots\frac{\partial{}T(x)}{\partial{}x_n}\right]\right|_{x^{*}}x(k)\nonumber\\
      =&\left.\left[\frac{\partial{}b(x)}{\partial{}x_1}\cdots\frac{\partial{}b(x)}{\partial{}x_n}\right]\right|_{x^{*}}x(k)\nonumber\\
      =&Hx(k).
      \label{eq:Linearization}
\end{align}
The local stability of a fixed point $x^{*}$ according to the eigenvalues ${\{\mu_i\}}_{i=1}^{n}$ from the spectral decomposition of the matrix $H$ are given by,
\begin{itemize}
    \item if $|\mu_i|<1$ for all $i=1,\ldots,n$ then $x^{*}$ is asymptotically stable,
    \item if $|\mu_i|>1$ for all $i=1,\ldots,n$ then $x^{*}$ is unstable,
    \item if $|\mu_i|>1$ for some $i=1,\ldots,n_{k}$ and $|\mu_i|<1$ for some $i=n_k+1,\ldots,n$, then $x^{*}$ is also unstable but has modal components that converge to it, making it a saddle point. 
\end{itemize}
\end{definition}
\begin{remark}
Note that the inequalities are strictly greater-than or less-than, this is in accordance with the hyperbolicity assumption. 
\end{remark}

With the information about the location and stability of fixed points, the main result of this paper is formulated, i.e., the approximation of the boundary of the ROA via eigenfunctions of the discrete-time Koopman operator. 
\subsection{Approximation of the ROA boundary}
\label{sub:Unit_Eig}

This section describes the approximation of the ROA of asymptotically stable fixed points in a data-driven approach using the discrete-time approximation of the Koopman operator. The claim of the main result of this study is: the level sets of a unitary eigenfunction give the boundary of the ROA of the asymptotically stable fixed points. This holds from the following set of premises.  
\begin{enumerate}
  \item A system that admits a Koopman operator transformation has an infinite set of eigenfunctions.%
    \label{it:infiniteEigenfunctions}
  \item Some nontrivial eigenfunctions have an associated eigenvalue equal to one (unitary eigenfunctions).%
    \label{it:nontrivialExists}
  \item A unitary eigenfunction is invariant along the trajectories of the system (from equation~\eqref{eq:PhiEvol}).%
    \label{it:unitaryIsInvariant}
  \item The trajectories of the system are level sets of a unitary eigenfunction.%
    \label{it:levelTrajectories}
  \item The stable manifold of a {\it type-one\/} saddle point in an $n$-dimensional system is an $(n-1)$-dimensional hyper surface composed by the union of the trajectories that converge to the point (from equation~\eqref{eq:stableM}).
  \item The level set of a unitary eigenfunction at a {\it type-one\/} fixed point is the stable manifold of the point.
  \item The boundary of the ROA of an asymptotically stable fixed point is the union of the stable manifolds of the {\it type-one\/} fixed points in the boundary (from \cref{th:AttractionRegions}).
\end{enumerate}

In other words, if it is possible to capture the $(n-1)$-dimensional hyper surfaces that converge to the {\it type-one\/} saddle points in the boundary of the ROA via the Koopman operator, then, it is possible to get the boundary of the ROA. If there are eigenfunctions that are invariant along the trajectories of the system, then, evaluating these eigenfunctions on the {\it type-one\/} saddle points of the system gives the constant values for which the unitary eigenfunctions capture the stable manifold of the saddle points, i.e., the level sets of the unitary eigenfunction. A level set of an arbitrary function $h(x)\colon\mathbb{R}^n\rightarrow\mathbb{C}$ for any constant $c\in\mathbb{C}$ is $\Gamma_{c}(h(x))=\{x\in{}\mathbb{R}^n\colon{}h(x)=c\}$. 

To guarantee the existence of nontrivial unitary eigenfunctions, consider the following property of the eigenfunctions of the Koopman operator (see property 1 in ~\cite{Mauroy2016}, for its analogous in continuous time).
\begin{property}
Suppose $\phi_1,\phi_2\in\mathcal{F}_d$ are Koopman eigenfunctions of an arbitrary system with associated eigenvalue $\mu_1$ and $\mu_2$. If  $\bar{\phi}(x)=\phi_{1}^{k_1}(x)\phi_{2}^{k_2}(x)\in\mathcal{F}_d$, for constants $k_1,k_2\in{}\mathbb{C}$, then, $\bar{\phi}(x)$ is an eigenfunction constructed upon the product of two arbitrary eigenfunctions with associated eigenvalue $\bar{\mu}=\mu_1^{k_1}\mu_2^{k_2}$.
\begin{proof}
  It follows from condition~\eqref{eq:PhiEvol} that 
\begin{align}
        \left[U_d\bar{\phi}\right](x)&=\left[U_d(\phi_1^{k_1}\phi_2^{k_2})\right](x)\\
                          &={\left([U_d\phi_1](x)\right)}^{k_1}{\left([U_d\phi_2](x)\right)}^{k_2}\\
                          &=\mu_1^{k_1}\phi_1^{k_1}(x)\mu_2^{k_2}\phi_2^{k_2}(x).
        \label{eq:eigenRelation}
    \end{align}
\end{proof}
\end{property}
This means that from a finite approximation of the Koopman operator, there is an infinity of possibly dependent eigenfunctions, showing that premise~\ref{it:infiniteEigenfunctions} holds. Moreover, from the previous construction, it is possible to find the complex constants $k_1$ and $k_2$ for eigenvalues $\mu_1$ and $\mu_2$ such that $\bar{\mu}=1$ and $\bar{\phi}(x)=\phi_{+}(x)$, where $\phi_+$ denotes an eigenfunction with a unitary associated eigenvalue. From~\eqref{eq:PhiEvol} it follows that $\phi_+(x)$ is an eigenfunction invariant along the trajectories of the system because
\begin{align}
  [U_d^k\phi_+](x)&=\phi_+(x)
\label{eq:constantPhi}
\end{align}

which demonstrates that premises~\ref{it:nontrivialExists} and~\ref{it:levelTrajectories} also hold and explains premise~\ref{it:unitaryIsInvariant}. As a result from the previous development, we can state our main result:

\begin{theorem}%
    \label{prop:Manifold}
     If there exists a dynamical system $(\mathcal{M};T(x);k)$ with multiple stable points that admits a Koopman operator approximation $(\mathcal{F}_d;U_{d};k)$, then the approximation of the stable manifold of the {\it type-one\/} saddle points $\hat{x}^*$ in the boundary of the ROA is the level set of a unitary eigenfunction $\phi_+(x)$ whose constant value comes from the evaluation of the unitary eigenfunction at the saddle point in the boundary, i.e.,     
    \begin{align}
       W^{s}(\hat{x}^{*})&=\Gamma_{c_{\phi_{+}(\hat{x}^*)}}\phi_+(x)\\
       &\subseteq\{x\in{}\mathbb{R}^n\colon{}\phi_{+}(x)=c_{\phi_{+}(\hat{x}^*)}\}.
        \label{eq:LevelSet}
    \end{align}
\end{theorem}
\begin{proof}
From the definition of stable manifold of a {\it type-one\/} saddle point~\eqref{eq:stableM} in \Cref{sub:ROA} 
\begin{equation}
  W^s(\hat{x}^{*})=\{x\in{}\mathbb{R}^{n}\colon{}\lim_{k\rightarrow{}\infty}T^{k}(x)=\hat{x}^{*}\}.
\end{equation}
Evaluating this manifold with an arbitrary eigenfunction $\phi(x)$, and using the condition of the Koopman operator in~\eqref{eq:Koopman} gives
\begin{align}
  \phi(W^s(\hat{x}^{*}))&\approx\{x\in{}\mathbb{R}^n\colon{}\lim_{k\rightarrow{}\infty}\phi(T^k(x))=\phi(\hat{x}^{*})\}\\
  &\approx\{x\in{}\mathbb{R}^n\colon{}\lim_{k\rightarrow{}\infty}[U_d^{k}\phi](x)=\phi(\hat{x}^{*})\}
\end{align}
Hence, for this arbitrary eigenfunction, the associated eigenvalue must be unitary for the equality to hold. Moreover, from \eqref{eq:constantPhi} it is clear that the time evolution of an eigenfunction with a unitary associated eigenvalue is invariant along the trajectories of the system, or in other words, it is independent on time evolution. These two developments imply that   
\begin{equation}
  \phi(W^s(\hat{x}^{*}))=\{x\in{}\mathbb{R}^n\colon{}\phi_+(x)=c_{\phi_+(\hat{x}^{*})}\}.
  \label{eq:uniteigProof}
\end{equation}
Notice that the right-hand part of \eqref{eq:uniteigProof} is the definition of a level set for an eigenfunction with unitary associated eigenvalue $\phi_+(x)$. As the trivial eigenfunction $\phi_{\mu=1}=1$ belongs to the set of eigenfunctions with unitary associated eigenvalue, the left-hand side holds for a subset of the unitary eigenfunctions $\phi_+(x)$. Therefore, the stable manifold of a {\it type-one} saddle point in the boundary of the ROA is
\begin{equation}
  \Gamma_{c_{\phi_{+}(\hat{x}^*)}}\phi_{+}(x)\subseteq\{x\in{}\mathbb{R}^{n}\colon{}\phi_{+}(x)=c_{\phi_{+}(\hat{x}^*)}\}.
        \label{eq:EiglevelSet}
\end{equation}
\end{proof}

In summary, the spectral decomposition of the discrete-time approximation of the Koopman operator can be used to find the nontrivial unitary eigenfunction, which, evaluated at the {\it type-one\/} points, characterizes the ROA of asymptotically stable points. 

\subsection{Algorithm}
\label{sub:algorithm}
The approach presented in \cref{sub:EDMD,sub:equilibrium_points,sub:Unit_Eig} for obtaining the attraction regions of asymptotically stable fixed points is summarized in Algorithm~\ref{alg:AtractionKoopman}, for which the following assumptions hold. 
\begin{enumerate}[label=B\arabic*:]
    \item The system under consideration has multiple hyperbolic fixed points.
    \item At least one of the fixed points is asymptotically stable.
    \item There is enough snapshot data either from measurements of the real system or a numerical simulation for constructing a discrete approximation of the Koopman operator. This condition can be checked with the empirical error that is a comparison between the data from the numerical integration of the system dynamics, and the state evolution map from the approximation of the discrete-time Koopman operator.
\end{enumerate}

\label{alg:AtractionKoopman}
\begin{algorithmic}[1]
  \Ensure{}Assumptions (B1--B3) are satisfied
  \State{\textbf{Data:} ${\{x_i,y_i\}}_{i=1}^{N}$}
  \Return{}$\partial{}R_A(x^{*}_{s})$
  \Procedure{ROA}{${\{x_i,y_i\}}_{i=1}^{N}, q, p$}
    \State{$[\Phi(x),M]\longleftarrow{}\text{EDMD}\left({\{x_i,y_i\}}_{i=1}^{N}, q, p\right)$}
    \State{$b(x)\longleftarrow\Psi_{B}^{-1}(B^\top\Xi{}M^k\Phi(x))$}
    \State{$x^{*}\longleftarrow \min_{x}\left \Vert{}b(x)-x\right \Vert^2_2$}
    \State{$E_{q}\longleftarrow{}\text{size}(x^{*})$}
    \For{$i\leftarrow{}1$, $E_{q}$}
      \State{$H_i\longleftarrow\left.\left[\frac{\partial{}b(x)}{\partial{}x_1}\cdots\frac{\partial{}b(x)}{\partial{}x_n}\right]\right|_{x^{*}_{i}}$}
      \State{${\{\mu_{i,j}\}}_{j=1}^{n}\longleftarrow{}\text{eig}(H_{i})$}
      \If{$|\mu_{i,j}|<1$ \text{\upshape for all} $j=1,\ldots,n$}
        \State{$x_{i}^{*}\leftarrow{}x^{*}_{s}$}
      \ElsIf{$|\mu_{i,j}|<1$ \text{\upshape for some} $j=1,\ldots,n$}
        \State{$x_{i}^{*}\leftarrow{}\hat{x}^{*}$}
      \Else{}
        \State{$x^{*}\leftarrow{}x_{u}^{*}$}
      \EndIf{}
    \EndFor{}
    \State{$\hat{E}_{q}\longleftarrow{}\text{size}(\hat{x}^{*})$}
    \State{$E_{s}\longleftarrow{}\text{size}(x_{s}^{*})$}
    \State{$\phi_{+}(x)\longleftarrow\{\phi_{1}^{k_1}(x)\phi_{2}^{k_2}(x)\colon{}\mu_1^{k_1}\mu_2^{k_2}=1\}$}
    \For{$i\leftarrow{}1$, $E_{s}$}
      \For{$j\leftarrow{}1$, $\hat{E}_q$}
        \State{$W_{j}^{s}(\hat{x}_{j}^{*})\longleftarrow\{x\in{}\mathbb{R}^{n}\colon{}\phi_{+}(x)=\phi_{+}(\hat{x}_{j}^{*})\}$}
      \EndFor{}
      \State{$\partial{}R_A(x^{*}_{i})=\cup{}W^{s}(\hat{x}^{*}_j)$}
    \EndFor{}
  \EndProcedure{}
\end{algorithmic}


\section{Simulation Results}
\label{sec:simulation_results}

For testing the reliability of the methodology and algorithm, they are applied to a model of competitive exclusion with two state variables. This model has the advantage of having an analytical description of the ROA under certain parameters, and is suitable to graphically show the effect of the eigenfunction with unitary eigenvalue. The reliability of the algorithm in a higher dimensional system composed by five state variables is then demonstrated, i.e., a mass-action kinetics (MAK) model. This latter example shows the advantage of not having to calculate level sets or handle complex geometries of the (n-1)-dimensional stable manifold hyper-planes. 
The population model and the mass-action kinetics systems are suitable for the analysis because of their geometric properties. They are non-negative compartmental systems that, depending on the parameterization, have hyperbolic fixed points that satisfy the Hartman-Grobman theorem~\cite{Chellaboina2009a}. The hyperbolicity of the fixed points implies that their unstable and stable manifolds intersect transversely at the saddle points. Therefore, Assumptions (A1--A3) are satisfied in this kind of systems under the right parameterization.

The numerical integration of the systems from a random set of initial conditions gives the dataset of trajectories for the algorithm. A subset is used to calculate the Koopman operator's approximation via the EDMD algorithm (training set), and another subset is used for testing the algorithm's accuracy. The algorithm's empirical error criteria for a number $N_s$ of test trajectories and a length $K_i$ for every one of those trajectories is
\begin{equation}
	e=\frac{1}{\sum_{i=1}^{N_s}K_i}\left(\sum_{i=1}^{N_s}\sum_{k=1}^{K_i}\frac{\vert{}T^k(x_i)-U_d^kg(x_i)\vert{}}{\vert{}T^k(x_i)\vert}\right),
\end{equation} 
where this error serves to determine the best $p-q$ parameters from a sweep over the different values and to verify whether assumption B3 holds.
\subsection{A Competition Model}
\label{ssub:LV}
 Consider the following network describing the dynamics of a population model where two species compete for the same resource.
 \begin{equation}
\begin{tikzpicture}[->,>=stealth',every node/.style={text height=1.5ex,minimum width=10ex}]
  \matrix (m) [matrix of math nodes, row sep=.1cm, column sep=.8cm, column 1/.style={nodes={align=left}}, column 2/.style={nodes={align=center}},column 3/.style={nodes={align=left}}]%
{$s_1$&&$2s_1$\\%
$s_2$ &&$2s_2$\\%
$s_1$+$s_2$&&$s_1$\\%
$s_1$+$s_2$&&$s_2$\\};
  \path
  ([yshift=.5mm]m-1-1.east) edge [above] node[yshift=-.9mm] {\scriptsize$r_{1}$} ([yshift=.5mm]m-1-3.west) 
 ([yshift=-.5mm]m-1-3.west) edge [below] node[yshift=1.5mm] {\scriptsize$r_{2}$} ([yshift=-.5mm]m-1-1.east) 
  ([yshift=.5mm]m-2-1.east) edge [above] node[yshift=-.9mm] {\scriptsize$r_{3}$} ([yshift=.5mm]m-2-3.west) 
 ([yshift=-.5mm]m-2-3.west) edge [below] node[yshift=1.5mm] {\scriptsize$r_{4}$} ([yshift=-.5mm]m-2-1.east) 
  ([yshift=.5mm]m-3-1.east) edge [above] node[yshift=-.9mm] {\scriptsize$r_{5}$} ([yshift=.5mm]m-3-3.west) 
  ([yshift=.5mm]m-4-1.east) edge [above] node[yshift=-.9mm] {\scriptsize$r_{6}$} ([yshift=.5mm]m-4-3.west); 
 \end{tikzpicture}%
 	\label{eq:LV_Network}
 \end{equation}
 where $s_1$ and $s_2$ are the competing species. The dynamics of the species depends on six parameters: $r_1$ and $r_3$ describe the population growth rates of each species respectively, $r_2$ and $r_4$ represents the logistic terms accounting for the competition between members of the same species (resulting into carrying capacities of the environment $r_1/r_2$ and $r_3/r_4$, respectively), $r_5$ describes the competitive effect of species $s_2$ on species $s_1$, and $r_6$ describes the opposite competitive effect between species. The values of these constants determine the potential outcome of the competition. Depending on their values, there can be a co-existence or an exclusion of one of the species against the other. The case of interest is the exclusion one, and the objective is to find the set of initial conditions within the state space that lead the model to one of the stable points where one species completely takes over the other. The differential equations that describe the population~\eqref{eq:LV_Network} are,
 \begin{align}
  	\dot{x}_1&=r_1x_1-r_2x_1^2-r_5x_1x_2%
    \label{eq:LV_DiffEq1}\\
  	\dot{x}_2&=r_3x_2-r_4x_2^2-r_6x_1x_2,
  	\label{eq:LV_DiffEq2}
  \end{align} 
  where the parameterization for the species to have two particular asymptotically stable fixed points is $r={[2\;1\;2\;1\;3\;3]}^{\top}$. With this choice, the system has indeed four fixed points: an unstable point at the origin defined as $x_{A}^{*}$, a saddle point at $(0.5, 0.5)$ defined as $x_{D}^{*}$ and two asymptotically stable at $(0,2)$ and $(2,0)$ defined as $x_{B}^{*}$ and $x_{C}^{*}$ respectively. The geometry of this problem provides a simple representation of the {\it type-one\/} saddle stable manifold which is the line $x_1=x_2$, therefore providing a closed formulation for the comparison with the stable manifold generated by the construction of the eigenfunction with an associated eigenvalue equal to one. 

  The numerical integration of the system from 200 uniformly distributed random initial conditions, and a $\Delta{}t=0.1$ over $x_1,x_2\in[0\;2]$ gives the datasets for approximating and testing the operator via the EDMD algorithm. Each set has 50\% of the available trajectories, that adds up to $7,285$ data-points each. 

  The application of the algorithm on the training set with a Laguerre polynomial basis, and a truncation scheme sweeping over several $p-q$ values gives the best performance when $q=1.1$ and $p=3$. This selection produces a polynomial basis of 13 elements of an order equal or less than 3 and a Koopman operator of order 13. Figure~\ref{fig:figure1} shows the retained indices after implementing the truncation scheme.
\begin{figure}[ht]
  \centering
  \includegraphics[width=0.45\textwidth]{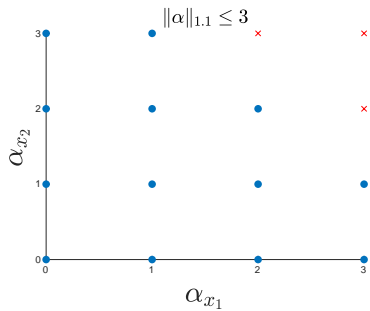}
  \caption{Retained indices for the approximation of the Koopman operator for the competition model, with a choice of $q=1.1$ and $p=3$.}%
  \label{fig:figure1}
\end{figure}

The next step of the method is the location and stability of the fixed points. Their location via \cref{lm:EquilibriumPoints} gives an absolute error of 0.15\%. Moreover, the method in \cref{sub:stability_of_equilibria} accurately provides their stability according to the  linearization of the nonlinear state evolution map~\eqref{eq:ForwardFlow} and the analysis of the eigenvalues $\lambda$ from the linearization evaluated at the identified fixed points. Table \ref{tab:NormLamDE} summarizes these results.

\begin{table*}[bht]
\centering
\caption{Competition model fixed points, location, and stability.}
\begin{tabular}{lccccc}
\toprule
          &Theoretical   & Algorithmic       & $|\lambda_1|$ & $|\lambda_2|$ & Stability \\ 
\midrule
$x_A^{*}$ & (0, 0)         & (-0.006, -0.006)   & 1.21          & 1.22          & Unstable  \\
$x_B^{*}$ & (0, 2)         & (0, 2)             & 0.66          & 0.82          & AS        \\
$x_C^{*}$ & (2, 0)         & (2, 0)             & 0.82          & 0.66          & AS        \\
$x_D^{*}$ & (0.5, 0.5)     & (0.5, 0.5)         & 0.81          & 1.10          & Saddle    \\ 
\bottomrule
\end{tabular}%
\label{tab:NormLamDE}
\end{table*}

Figure~\ref{fig:LVT} depicts the comparison between the theoretical trajectories given by the integration of the differential equations~\eqref{eq:LV_DiffEq1}--\eqref{eq:LV_DiffEq2} and the state evolution map~\eqref{eq:ForwardFlow} from the same initial conditions. It also depicts the comparison between the theoretical boundary of the attraction regions $x_1=x_2$ and the boundary given by the level set of the constructed unitary eigenfunction from~\eqref{eq:LevelSet}. The error in the classification of the initial conditions is of $2\%$, while the mean absolute error between the boundary of the regions of attraction is $3\%$.

\begin{figure}[ht]
	\centering
	\includegraphics[width=\textwidth]{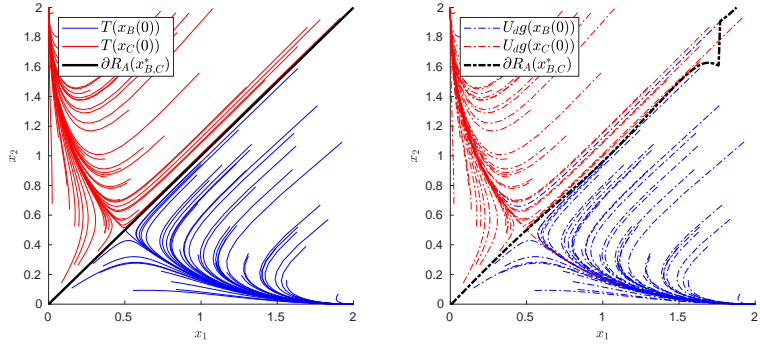}
	\caption{Trajectories and boundary of the asymptotically stable points of a. system differential equations, and b. Koopman operator and the unitary eigenfunction.}%
	\label{fig:LVT}
\end{figure}

Figure~\ref{fig:LVE} shows four eigenfunctions of the Koopman operator, including the trivial eigenfunction whose value is constant in all the state space. This eigenfunction is commonly the result of the EDMD algorithm and does not provide any useful information about the system. Therefore, it is necessary to get a nontrivial unitary eigenfunction from~\eqref{eq:eigenRelation}. To achieve this objective, setting $\mu_1^{k_1}\mu_2^{k_2}=1$ and computing a solution for $k_1$ and $k_2$ gives the desired unitary eigenfunction. However, there is a caution with this calculation as there is an infinite number of solutions for the choice of constants. In practice, the method is to set one and calculate the other.

Figure~\ref{fig:LVE} also depicts the resulting unitary eigenfunction that captures the stable manifold of the saddle point, and the two eigenfunctions with real-valued eigenvalues close to one used for the construction. The eigenvalues of these eigenfunctions are $\mu_{1}=1.07$ and $\mu_{2}=0.83$. Inspecting the near-unitary eigenfunction with $\mu_{1}=1.07$ shows that it could potentially be used for the classification on its own, and indeed, this eigenfunction does give good results, although not as accurate as the construction of the unitary eigenfunction. It is still not clear how to select the base eigenfunctions properly for the construction of the unitary eigenfunction, and although the algorithm gives accurate results when the base eigenfunctions for the construction of the unitary eigenfunction do not have real-valued eigenvalues, not having real-valued eigenfunctions for the construction does hinder the accuracy.

\begin{figure}[ht]
	\centering
  \includegraphics[width=\textwidth]{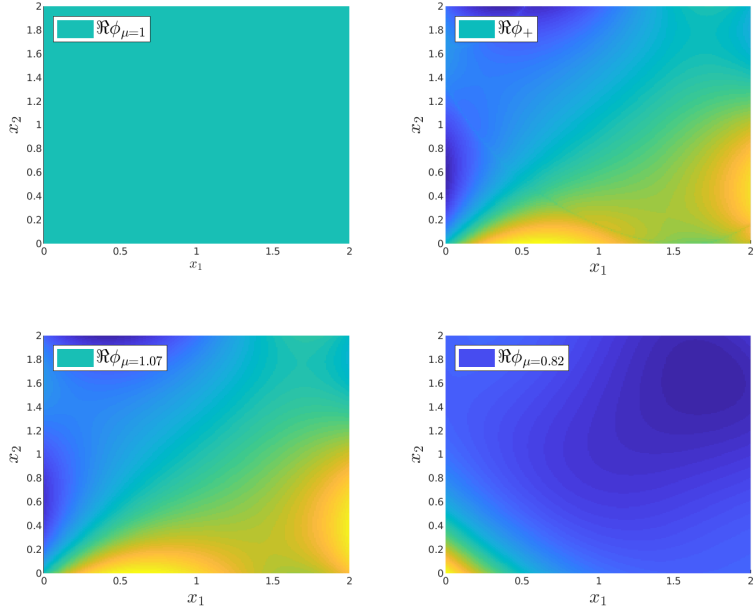}
  \caption{Eigenfunctions of the Koopman operator: a. Trivial eigenfunction with $\mu=1$, b. Constructed eigenfunction $\phi_{+}$, c. First eigenfunction for constructing $\phi_{+}$, and d. Second eigenfunction for constructing $\phi_{+}$.}%
  \label{fig:LVE}
\end{figure}

\subsection{Mass Action Kinetics}
\label{sub:mass_action_kinetics}
Consider the following network that describes an auto-catalytic replicator in a continuous flow stirred tank reactor 
\begin{equation}
\begin{tikzpicture}[->,>=stealth']
  \matrix (m) [matrix of math nodes, row sep=.3cm, column sep=.75cm, column 1/.style={nodes={align=right}}, column 2/.style={nodes={align=center}},column 3/.style={nodes={align=left}}]{$s_1$+$2s_3$&&$s_2$+$3s_3$\\$s_2$+$2s_4$&&$3s_4$\\$s_3$&&$s_5$\\$s_4$&&$s_5$\\};
  \path
  ([yshift=.5mm]m-1-1.east) edge [above] node[yshift=-.9mm] {\scriptsize$r_{1}$} ([yshift=.5mm]m-1-3.west) 
  ([yshift=.5mm]m-2-1.east) edge [above] node[yshift=-.9mm] {\scriptsize$r_{2}$} ([yshift=.5mm]m-2-3.west) 
  ([yshift=.5mm]m-3-1.east) edge [above] node[yshift=-.9mm] {\scriptsize$r_{3}$} ([yshift=.5mm]m-3-3.west) 
  ([yshift=.5mm]m-4-1.east) edge [above] node[yshift=-.9mm] {\scriptsize$r_{4}$} ([yshift=.5mm]m-4-3.west); 
 \end{tikzpicture}
\label{eq:AutocatR}
\end{equation}
where there are two species, $s_3$ and $s_4$. The species $s_3$ consumes the substrate  $s_1$ to reproduce, and as a byproduct, produces the resource  $s_2$ which is suitable for the reproduction of the second species $s_4$. Finally, $s_5$ is the dead species from both groups. The constants $r_1>r_3$ and, $r_2>r_4$ are the pairs of replication rate and the species death. Considering $d$ as the dilution rate, or the material exchange with the environment, the dynamics of the network~\eqref{eq:AutocatR} is described by the differential equations
\begin{align}
			\dot{x}_1&=-r_{1}x_{1}x_{3}^2+d-dx_{1}%
      \label{eq:SimpleDynamics1}\\ 
			\dot{x}_2&=+r_{1}x_{1}x_{3}^2-r_{2}x_{2}x_{4}^2-dx_{2}\\ 
			\dot{x}_3&=+r_{1}x_{1}x_{3}^2-r_{3}x_{3}-dx_{3}\\
			\dot{x}_4&=+r_{2}x_{2}x_{4}^2-r_{4}x_{4}-dx_{4}\\
			\dot{x}_5&=+r_{3}x_{3}+r_{4}x_{4}-dx_{5},
			\label{eq:SimpleDynamics5}	
\end{align}
where substrate $s_{1}$ is the only component in the input flow. The value for the reaction rates vector is $k={[7\:5\:0.3\:0.05]}^{\top}$, yielding five fixed points: three asymptotically stable (the working point defined as $x_{A}^{*}$ where the two species thrive and coexist, a point where species $s_3$ thrives and species $s_4$ washes-out defined as $x_{C}^{*}$, and a wash-out point where the concentration of both species is zero defined as $x_{E}^{*}$), and two saddle points defined as $x_{B}^{*}$ and $x_{D}^{*}$. The objective is to find the ROA of the working point using the unitary eigenfunction. This example highlights the contributions of this work, as the dimension of the state is 5, instead of being a two-dimensional one or a three-dimensional system that can be sliced for the analysis. The stable manifold of saddle points is a 4-dimensional hyper surface. Therefore, it has a complex geometry to analyze, and the criterion for classifying the test set of initial conditions is not trivial. Particularly, we show that it is not necessary to deal with complex geometries or higher dimensional hyper-planes to perform the classification. 

The numerical integration of the differential equations~\eqref{eq:SimpleDynamics1}--~\eqref{eq:SimpleDynamics5} from 360 randomly distributed initial conditions and $\Delta{}t=0.1$ generates the set of orbits for training and testing the algorithm. From this set of orbits, $50\%$ are used to approximate the operator via the EDMD algorithm, and the other $50\%$ to test the accuracy of the state evolution~\eqref{eq:ForwardFlow} map, and the accuracy of the classification, where each set has $79,948$ points. Additionally, sweeping over several $p-q$ values gives the best performance for the truncation scheme when $q=0.8$ and $p=4$ which leads to an approximation of the Koopman operator of order 163, and thus, a set of 163 eigenfunctions, eigenvalues, and modes. From the set of eigenfunctions, there are two eigenfunctions with real eigenvalue closest to one, which are $\mu_1=1.00008$ and $\mu_2=0.99983$. Given that these eigenvalues associated with these eigenfunctions are close enough to one, these unitary eigenfunctions are sufficient to perform the analysis.

The identification of the fixed points via \cref{lm:EquilibriumPoints} gives an absolute error of 0.15\%, and their stability according to the method in \cref{sub:stability_of_equilibria} provides an accurate description. Table~\ref{tab:FiveEqLoc} shows the results from the theoretical and algorithmic location of the fixed points, Table~\ref{tab:FiveStab} shows the results of taking the norm of the eigenvalues $\lambda$ from the Hessian matrix of~\eqref{eq:ForwardFlow} evaluated at each of the previously identified fixed points. 

\begin{table*}[ht]
\centering
\caption{Location of fixed points.}%
\label{tab:FiveEqLoc}
\begin{tabular}{lccc}
\toprule
           & Theoretical                & Algorithmic                & Error \%  \\ 
\midrule
$x_A^{*}$  & (0.23, 0.09, 0.30, 0.54, 0.59) & (0.23, 0.09, 0.30, 0.54, 0.59) & 0.08   \\
$x_B^{*}$  & (0.21, 0.67, 0.30, 0.07, 0.47) & (0.23, 0.62, 0.30, 0.11, 0.49) & 0.00   \\
$x_C^{*}$  & (0.23, 0.76, 0.30, 0.00, 0.46) & (0.23, 0.76, 0.30, 0.00, 0.46) & 0.00   \\
$x_D^{*}$  & (0.76, 0.23, 0.09, 0.00, 0.14) & (0.70, 0.30, 0.11, 0.00, 0.17) & 0.54   \\
$x_E^{*}$  & (1.00, 0.00, 0.00, 0.00, 0.00) & (1.00, 0.00, 0.00, 0.00, 0.00) & 0.00   \\
\bottomrule
\end{tabular}
\end{table*}
\begin{table*}[ht]
\centering
\caption{Stability of fixed points.}%
\label{tab:FiveStab}
\begin{tabular}{lcclllcl}
\toprule
           & $|\lambda_1|$  & $|\lambda_2|$  & $|\lambda_3|$  & $|\lambda_4|$  & $|\lambda_5|$  & Stability &                  \\ 
\midrule
$x_A^{*}$  & 0.89           & 0.98           & 0.98           & 0.97           & 0.98           & AS        & Working Point    \\
$x_B^{*}$  & 1.04           & 0.98           & 0.98           & 0.98           & 0.98           & Saddle    &                  \\
$x_C^{*}$  & 0.98           & 0.98           & 0.93           & 0.98           & 0.98           & AS        & $x_4$ Wash-out   \\
$x_D^{*}$  & 1.05           & 0.96           & 0.98           & 0.98           & 0.97           & Saddle    &                  \\
$x_E^{*}$  & 0.90           & 0.98           & 0.98           & 0.97           & 0.97           & AS        & Wash-out         \\
\bottomrule
\end{tabular}
\end{table*}

The next step in the algorithm is finding a classification scheme that does not depend on the geometry of the 4-dimensional hyperplanes that compose the stable manifold of {\it type-one\/} saddle points. For this purpose, we use a saddle classifier~\cite{Garcia-Tenorio2021b}. This type of classifier divides the set of initial conditions $x_0\in\mathcal{M}$ of the testing trajectories into a subset $X\subset\mathcal{M}$ and its complement $X'$, where the evaluation of an arbitrary initial condition with a unitary eigenfunction compared to the evaluation of the {\it type-one\/} point with the same unitary eigenfunction gives the criterion for the division
\begin{equation}
  \label{eq:unitaryDivision}
  X\subseteq\left\{x_0\in\mathbb{R}^{n}\colon\Re{\phi_+(x_0)}\geq\Re{\phi_{+}(\hat{x}^*)}\right\},
\end{equation}
where $\Re$ denotes the real part of the evaluation. This gives a simple algebraic criterion for the classification of initial conditions into their respective ROA\@. 

\Cref{fig:EigClass} depicts the results of performing the evaluation of the saddle points and the initial conditions of the test set with the unitary eigenfunctions: $\phi_{+_1}$ corresponds to the eigenfunction with $\mu_{+_1}=1.00008$ and $\phi_{+_2}$ corresponds to the eigenfunction with $\mu_{+_2}=0.99983$. The test set is carefully chosen such that there are 60 initial conditions that converge to each of the three asymptotically stable fixed points. Furthermore, the division of this set into two subsets gives one set of 90 testing initial conditions to perform the analysis, and another set of 90 initial conditions to perform the final test of the classification. Indexing over the analysis set gives the horizontal axis in \cref{fig:EigClass}. Moreover, the vertical axis shows the result of evaluating each of the 90 different initial conditions with the trivial unitary eigenfunction, $\phi_{1_+}$, where it is clear that the trivial eigenfunction does not give any information. 

The other two unitary eigenfunctions do give some useful information about the classification scheme. Using $\phi_{+_1}$, it is clear that evaluating the saddle point $\hat{x}^*_B$ is an accurate criterion for the classification of initial conditions that converge to the asymptotically stable point $x_A^*$, i.e.,
\begin{equation}
  \label{eq:xAclassifier}
  X_A\subseteq\left\{x_0\in\mathbb{R}^5\colon\phi_{+_1}(x_0)\leq\phi_{+_1}(\hat{x}^*_B)\right\}.
\end{equation}

Furthermore, performing the same evaluation process with the second unitary eigenfunction $\phi_{+_2}$, it is clear that evaluating the saddle point $\hat{x}^*_D$ is an accurate comparison criterion to classify the initial conditions that converge to the wash-out point $x^*_E$, i.e.,
\begin{equation}
  \label{eq:xEclassifier}
  X_E\subseteq\left\{x_0\in\mathbb{R}^5\colon\phi_{+_2}(x_0)\leq\phi_{+_2}(\hat{x}^*_D)\right\}.
\end{equation}

Notice that the $\Re$ notation is dropped for the two classifiers~\eqref{eq:xAclassifier} and~\eqref{eq:xEclassifier} because these are real-valued eigenfunctions. The combination of these two conditions specifies the criteria to classify the initial conditions that converge to the different asymptotically stable fixed points. Figure~\ref{fig:ClassPlot} shows the result of applying the classification criteria over the second set of testing initial conditions. Again, the horizontal axis indexes over the three groups of 30 initial conditions that converge to each of the stable points. Moreover, the vertical axis shows the classification of the initial conditions according to the classifiers, where an incorrect classification has the correct value surrounding it. The evaluation of the points with the second unitary eigenfunction is only for illustration purposes. The final test over this set yields an error of $12\%$ miss-classified points.

\begin{figure*}[ht]
  \centering
  \includegraphics[width=\textwidth]{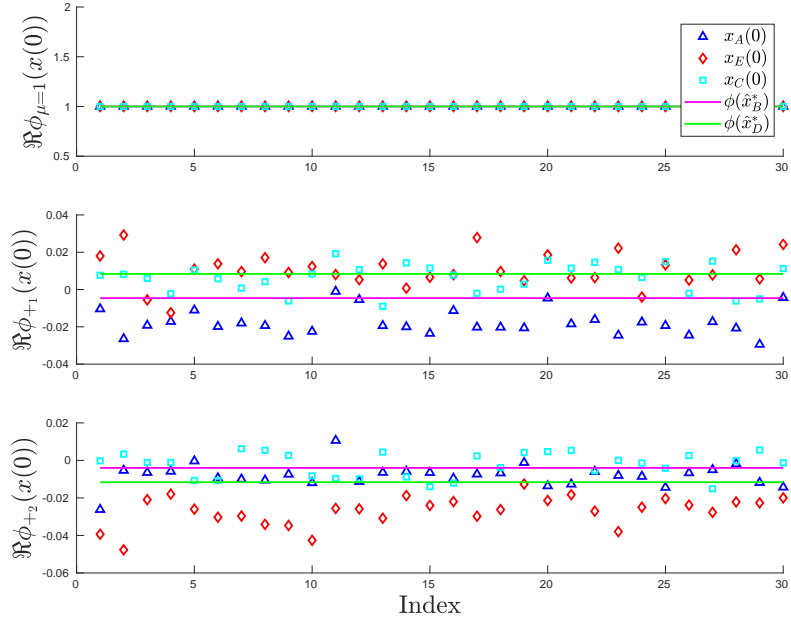}
  \caption{Eigenfunctions with unitary associated eigenvalue: a. Trivial b. $\mu>1$ c. $\mu<1$.}%
  \label{fig:EigClass}
\end{figure*}

\begin{figure*}[ht]
  \centering
  \includegraphics[width=\textwidth]{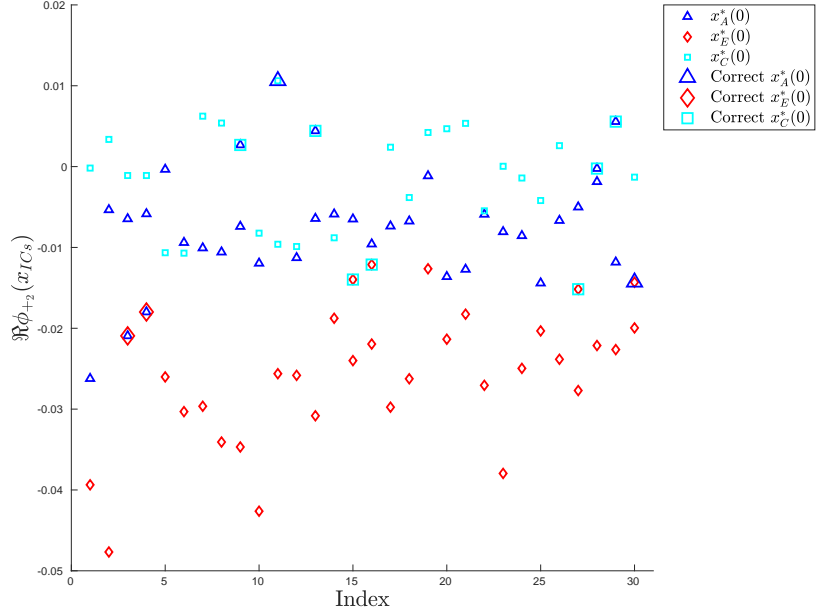}
  \caption{Classification of the initial conditions according to the evaluation of the eigenfunctions on the saddle points.}%
  \label{fig:ClassPlot}
\end{figure*}

\section{Critical Analysis and Perspectives}%
\label{sec:analysis}
Even if we had the tools for extracting all kinds of information regarding the system, there are limitations for the deduction of the eigenfunctions. The right choice of the polynomial basis based on the known structure of the differential equation, the optimal truncation scheme, the choice of eigenfunctions to construct the one with unitary associated eigenvalue are open questions for improving the algorithm and the analysis.

Furthermore, there is the problem of the amount and quality of the required data to approximate the Koopman operator accurately. For the two case studies, the amount of data is sufficient to produce an accurate nonlinear model of the systems, and the integration gives clean data devoid of noise. Moreover, to consider the proposed method over traditional identification techniques, it is necessary to develop experimental design techniques specifically for the calculation of the Koopman operator to reduce the necessary data. Additionally, the data from experimental setups is noisy, and often lacks the measurements of the whole state. Therefore, it is necessary to complement the algorithm with noise filtering observer designs to handle these issues.

Furthermore, the relation between the eigenfunctions and their associated eigenvalues is still open. There are other dynamical characteristics to be analyzed from this association, as the dynamic behavior of eigenfunctions whose associated eigenvalue matches the local spectra of the fixed points is a potential source of information of the underlying dynamical system.

A possible improvement would be for the choice or construction of eigenfunctions; all eigenfunctions capture different information and properties (relevant or not, accurate or not), and the analysis is limited to the invariant sets of these eigenfunctions. An analysis based on the spectral and geometric characteristics of real-valued and complex-valued eigenfunctions could give more information about the system for control purposes.

\section{Conclusions}
\label{sec:conclusions}
This paper deals with analysis aspects of nonlinear dynamical systems using the spectrum of the Koopman operator. Specifically, the construction of invariant sets called the unitary eigenfunctions that capture the stable manifold of {\it type-one\/} saddle points. The union of these stable manifolds is the boundary of the region of attraction of asymptotically stable fixed points, and the evaluation of these unitary eigenfunctions gives a simple algebraic criterion to determine the region of attraction of these stable points. Furthermore, all this information comes from a data-driven approach without the need for the differential equation model. Based solely on data from the system, the algorithm is able to provide the location and local stability of the fixed points of the system. The extension of these data-driven tools can provide new methods for analyzing complex systems and the ability to reformulate traditional linear controller synthesis methods to work with the discrete-time Koopman operator approximation. 

The data-driven representation of the Koopman operator provides the spectral decomposition that can complement traditional techniques for synthesis and control and provide a black-box tool for dealing with systems whose exact dynamics are challenging to model and identify.

To improve the analysis, we need to explore the polynomial structure of the system and find relationships between the system and the polynomial bases that give the approximation to optimize their selection. These relationships, coupled with more accurate eigenfunctions deduced from optimal truncation schemes, can allow the identification of eigenfunctions with additional information regarding the nonlinear dynamical system for control purposes and reduce the necessary amount and quality of the data for the approximation.



\bibliographystyle{ieeetr}
\bibliography{wileyNJD-AMA}
\end{document}